\documentclass[11pt,a4paper]{article}

\usepackage{amsmath}
\usepackage{amsthm}
\usepackage{amssymb}
\usepackage{amsfonts}
\usepackage{graphicx}
\usepackage{makeidx}
\makeindex

\usepackage{latexsym}
\usepackage{graphicx}
\usepackage{marvosym}
\usepackage{hyperref}
\usepackage{float}
\usepackage{color}
\def\rr{{\mathbb R}}

\def\lie{{\cal L}}

\def\f12{\frac 1 2}

\def\th{\theta}

\def\a{\alpha}
\def\b{\beta}
\def\ga{\gamma}

\def\hh{\mathcal{H}^{+}}
\def\m{\mathcal{M}}

\newcommand{\lapp}{\mbox{$\triangle \mkern-13mu /$\,}}
\newcommand{\epsi}{\mbox{$\epsilon \mkern-7.4mu /$\,}}
\newcommand{\gi}{\mbox{$g \mkern-8.8mu /$\,}}
\newcommand{\di}{\mbox{$d \mkern-9.2mu /$\,}}

\def\f12{\frac 1 2}

\newtheorem{remark}{Remark}[section]

\newtheorem{proposition}{Proposition}[section]

\newtheorem{mytheo}{Theorem}

\begin{document}
\title{Horizon Instability of Extremal Black Holes}

\author{Stefanos Aretakis\thanks{Princeton University, Department of Mathematics, Fine Hall, Washington Road, Princeton, NJ 08544, USA.}\thanks{  Institute for Advanced Study, Einstein Drive, Princeton, NJ 08540, USA.}}


\maketitle
\begin{abstract}
We show that axisymmetric extremal horizons are unstable under  scalar perturbations. Specifically, we show that translation invariant derivatives of generic solutions to the wave equation do not decay along such horizons as advanced time tends to infinity, and in fact, higher order derivatives blow up. This result holds in particular for extremal Kerr--Newman and Majumdar--Papapetrou spacetimes and is in stark contrast with the subextremal case for which decay is known for all derivatives along the event horizon.
\end{abstract}

\tableofcontents

\section{Introduction}
\label{sec:Introduction}

 \textit{Extremal} black holes are central objects of study for  high-energy physics and have attracted significant interest in the mathematical community on account of their elaborate analytical features. In this paper we will exhibit instability properties of a general class of extremal black holes with respect to scalar perturbations. Remarkably, these instabilities are completely determined by local properties of extremal horizons and hence do not depend on the global aspects of spacetime. The present work generalises previous results of the author \cite{aretakis1,aretakis2} on extremal Reissner--Nordstr\"{o}m backgrounds.  

\subsection{The Main Results}
\label{sec:TheMainResults}

Our general set-up is the following. We consider 4-dimensional $\mathbb{R}\times\mathbb{T}^{1}$-symmetric Lorentzian manifolds $(\m,g)$ containing an extremal horizon $\hh$. Specifically, we assume that there exist a Killing vector field $V$ which is normal to a null hypersurface $\hh$ such that the following extremality condition is satisfied:
\[\nabla_{V}V=0:\text{ on }\hh.\]
We will refer to $\hh$ as the horizon. Note that $\hh$ does not necessarily have to be the event horizon of a black hole region, i.e.~$\hh$  may be an isolated extremal horizon.  We also assume that there exists an additional axial Killing field $\Phi$ tangential to $\hh$ with closed orbits. The precise geometric assumptions on $(\m,g)$ are described in Section \ref{sec:Assumptions}.
\begin{figure}[H]
	\centering
		\includegraphics[scale=0.12]{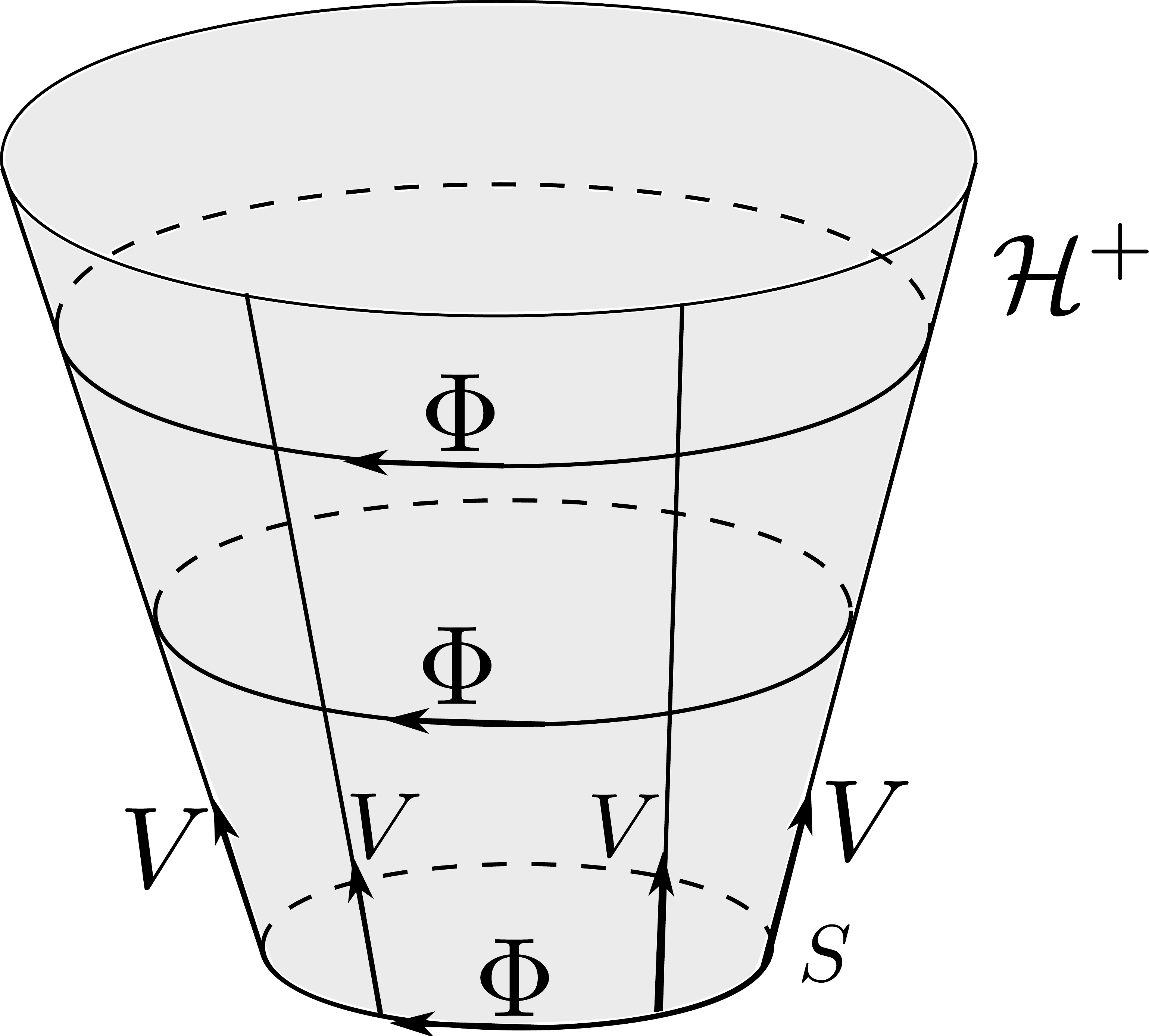}
	\label{fig:p421}
\end{figure}
We then study the behaviour of solutions to the \textit{wave equation}
\begin{equation}
\Box_{g}\psi=0.
\label{wave}
\end{equation}
We  prove that the first order derivatives generically \textbf{do not decay} along $\hh$, and in fact,  that the higher order derivatives asymptotically \textbf{blow up} along $\hh$ (Sections \ref{sec:ConservationLawAlongExtremalHorizon}, \ref{sec:NonDecayAndBlowUpAlongHh}). The genericity here refers to a condition on a section of $\hh$.
 
The source of the above instability results is a \textit{conservation law} that holds on $\hh$. This conservation law remarkably depends only on the local geometric properties of the horizon and not on global aspects of the spacetime. Hence we will not impose global hyperbolicity or discuss global well-posedness of the wave equation (see for example \cite{wellposedness}).

 These instabilities apply for Majumdar--Papapetrou multi black hole spacetimes (see Section \ref{sec:MajumdarPapapetrouMultiBlackHoles}) and extremal Kerr backgrounds allowing a cosmological constant $\Lambda\in\mathbb{R}$ (see Section \ref{sec:ExtremalKerr}). Moreover, for the latter backgrounds (with $\Lambda=0$) we also prove that the \textit{energy} of higher order derivatives generically blows up. 

The author has previously studied in \cite{sa10,aretakis1,aretakis2} the wave equation on a simpler model of extremal black holes, namely the spherically symmetric charged Reissner--Nordstr\"{o}m black holes. Solutions on such backgrounds were shown to exhibit both stability and instability properties. The analogues of the stability results for extremal Kerr were presented in \cite{aretakis3} for axisymmetric solutions. We describe these results in more detail in the next subsection where we also put them into context by briefly summarising previous mathematical work on the linear wave equation on black hole backgrounds.

\subsection{The Black Hole Stability Problem}
\label{sec:TheBlackHoleStabilityProblem}

The first step in investigating the dynamic stability (or instability) of a spacetime (see \cite{lecturesMD}) is by understanding the dispersive properties of the wave equation \eqref{wave}.

Work on the wave equation on black hole spacetimes began in 1957 for the Schwarzschild case with the pioneering work of Regge and Wheeler \cite{regge}, but the first complete quantitative dispersive result was obtained only in 2005 by Dafermos and Rodnianski \cite{redshift}, where the authors introduced a vector field that captures in a stable manner the so-called \textit{redshift effect} on the event horizon. The origin of their constructions lies in the positivity of the \textit{surface gravity} $\kappa$ given by \[\nabla_{V}V=\kappa V,\]
where $V$ is Killing and normal to the event horizon. During the last decade remarkable progress was made and the definitive understanding of decay on general \textit{subextremal} Kerr backgrounds was presented in \cite{tria} (see also \cite{damon, gusmu1,volker}). All these developments use in one way or another the redshift effect. For an exhaustive list of references see \cite{lecturesMD}.

 The analysis of the wave operator on extremal black holes was initiated in \cite{aretakis1,aretakis2} (see also \cite{sa10}), where definitive stability and instability results were obtained for extremal Reissner--Nordstr\"{o}m backgrounds. In these spacetimes, let $\Sigma_{0}$ be a spacelike hypersurface crossing $\hh$ and  $\Sigma_{\tau}=\varphi^{V}_{\tau}(\Sigma_{0})$, where $\varphi_{\tau}^{V}$ denotes the flow of $V$. Let also  $E_{\Sigma_{\tau}}[\psi]$ denote the energy measured by a local observer on $\Sigma_{\tau}$ and $Y$ denote a $\varphi_{\tau}^{V}$-invariant transversal to $\hh$ vector field.
\begin{figure}[H]
	\centering
		\includegraphics[scale=0.12]{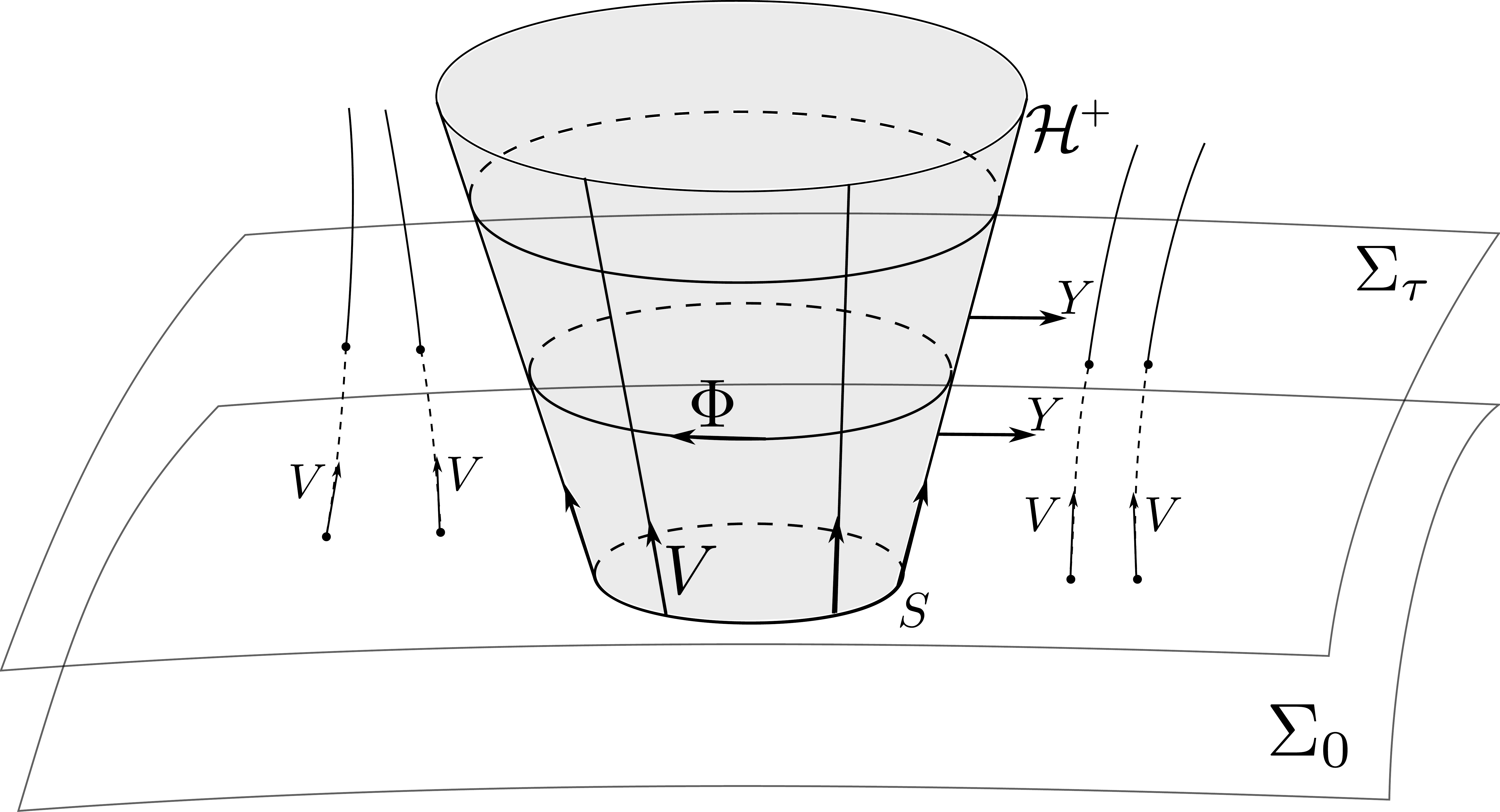}
	\label{fig:p3}
\end{figure} 
The main results, for extremal Reissner--Nordstr\"{o}m, of the analysis of \cite{aretakis1,aretakis2} include the following:

1. \textbf{Pointwise decay for $\psi$}: For all solutions $\psi$ to \eqref{wave} which arise from initial data with bounded energy we have $\left\|\psi\right\|_{L^{\infty}(\Sigma_{\tau})}\rightarrow 0$ as $\tau\rightarrow+\infty$.

2. \textbf{Decay of degenerate (at $\hh$) energy of $\psi$}: If the horizon $\hh$ corresponds to $r=M$, where $M>0$, then $E_{\Sigma_{\tau}\cap\left\{r\geq r_{0}>M\right\}}[\psi]\rightarrow 0$ as $\tau\rightarrow+\infty$. 

3. \textbf{Non-decay along $\hh$}: For generic $\psi$, we have that $|Y\psi|$ does not decay along $\hh$.

4. \textbf{Pointwise blow-up of higher order derivatives along $\hh$}: For generic $\psi$, we have $|Y^{k}\psi|\rightarrow+\infty$ along $\hh$, for all $k\geq 2$, as $\tau\rightarrow+\infty$.

5. \textbf{Energy blow-up of higher order derivatives:} For generic $\psi$, we have $E_{\Sigma_{\tau}}[Y^{k}\psi]\rightarrow +\infty$, for all $k\geq 1$, as $\tau\rightarrow+\infty$.

We remark that the latter non-decay and blow-up results are in sharp contrast with the non-extremal case where decay holds for all higher order derivatives of $\psi$ along $\hh$. Furthermore, we note that sharp quantitative estimates have been shown in \cite{aretakis1,aretakis2} for each angular frequency $l\in\mathbb{N}$. In order, however, to simplify the notation above, we have only presented the corresponding qualitative results. Results analogous to $1-2$ above have been shown in \cite{aretakis3} for axisymmetric solutions on extremal Kerr backgrounds. Note, however, that obtaining instability results for non-spherically symmetric extremal black holes (which is the topic of the present paper) had remained an open problem.

For other results regarding extremal black holes see the discussion in Section \ref{sec:TheVacuumAndElectrovacuumReduction}.

\section{The Geometric Set-up}
\label{sec:TheGeometricSetting}

\subsection{Assumptions}
\label{sec:Assumptions}

We consider 4-dimensional Lorentzian manifolds $(\m,g)$ which satisfy the following properties:
\begin{description}
	\item[\textbf{A1}.] $(\m,g)$ is $\mathbb{R}\times\mathbb{T}^{1}$-symmetric, i.e.~it admits two commuting Killing vector fields $V,\Phi$. The vector field $V$ has complete integral curves homeomorphic to the line $\rr$. The vector field $\Phi$ has closed spacelike integral curves and vanishes on a 2-dimensional submanifold $\mathcal{A}$, called the axis. 
	\item[\textbf{A2}.] There exists a null hypersurface $\hh$, which we will refer to as the horizon, such that $\Phi$ is tangential to $\hh$ and $V$ is normal to $\hh$ and also satisfies
	\begin{equation}
\nabla_{V}V=0: \text{ on }\hh.
\label{extrem}
\end{equation}
The condition \eqref{extrem} captures the extremality of $\hh$. 
\item[\textbf{A3}.] The topology of the horizon sections is spherical, i.e.~$\hh$ contains a spacelike 2-surface $S$ homeomorphic to the 2-sphere.
\item[\textbf{A4}.] The Killing fields $V,\Phi$ satisfy the Papapetrou condition, namely that the distribution of the planes orthogonal to the planes spanned by $V$ and $\Phi$ is integrable.
\end{description}
\begin{figure}[H]\centering
		\includegraphics[scale=0.135]{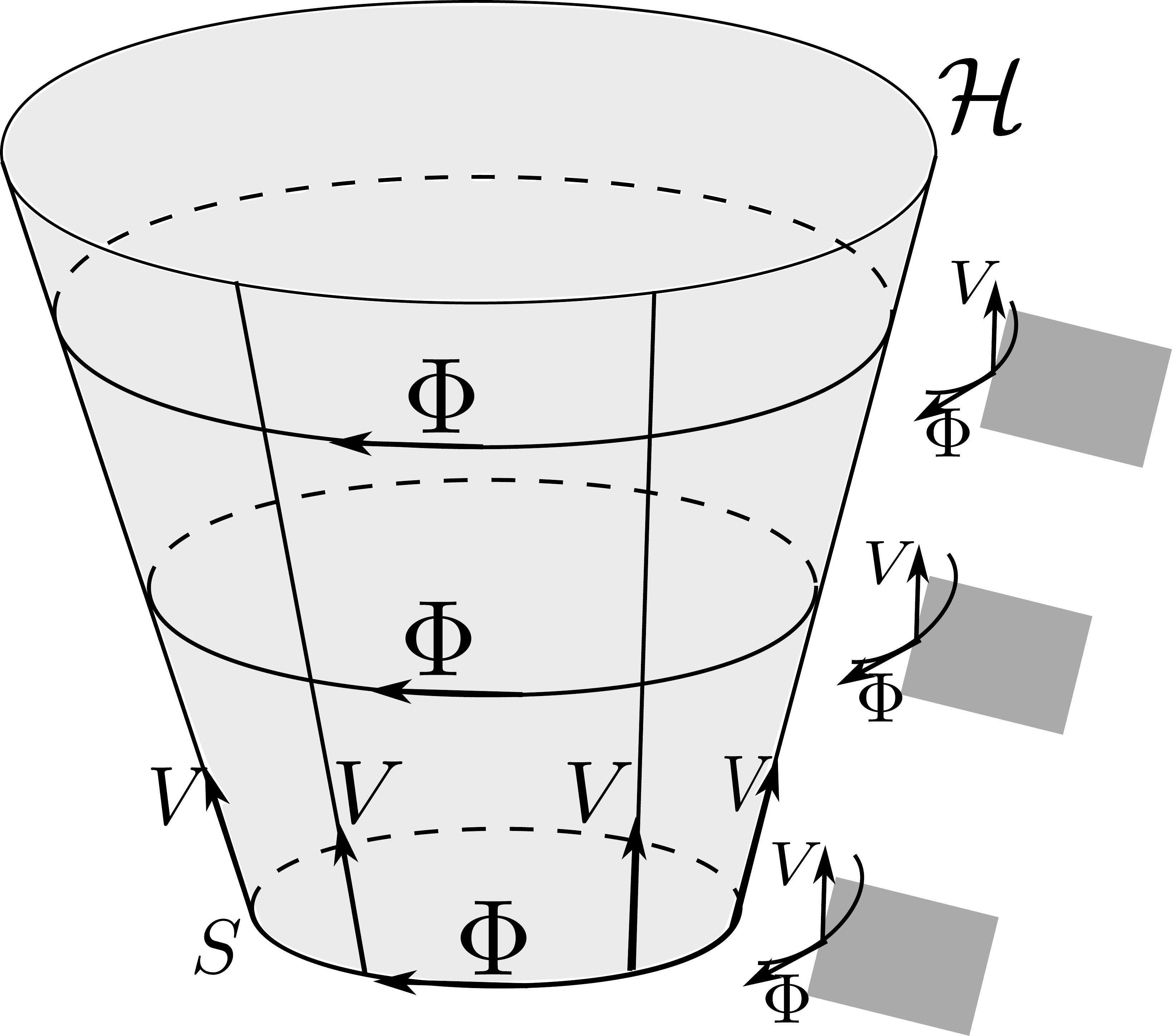}
	\label{fig:p4}
\end{figure}

\subsection{The Adapted Coordinate System}
\label{sec:TheAdaptedCoordinateSystem}

Under the assumptions A1--A4, we obtain the following
\begin{proposition}
Let $(\m,g)$ be a 4-dimensional Lorentzian manifold which satisfies the assumptions $A1-A4$ of Section \ref{sec:Assumptions}. Then, there exists a coordinate system $(v, r, \th,\phi)\in (-\infty,+\infty)\times (-\epsilon,\epsilon)\times (0,\pi)\times (0,2\pi),$ for some $\epsilon>0$ with $\hh=\left\{r=0\right\}$, $\partial_{v}=V, \partial_{\phi}=\Phi$ and such that if we denote $ Y=\partial_{r}, \Theta=\partial_{\theta}$ then
\begin{equation*}
\Theta\perp V,\ \ \ \  \Theta\perp\Phi,
\end{equation*}
\label{coordprop}
everywhere in the domain of the above system.
\end{proposition}
\begin{proof}
Let $\mathcal{A}$ denote the axis of the axisymmetric action, namely the set of points for which $\Phi=0$. This is a 2-dimensional timelike manifold. 

Let $\Sigma_{0}$ be an axisymmetric (i.e.~$\Phi$ is tangential to $\Sigma_{0}$) spacelike hypersurface crossing the horizon $\hh$. Since we can write (in a non-unique manner)  $\Sigma_{0}=(-\epsilon, \epsilon)\times S^{2}$ where $S^{2}$ are 2-dimensional axisymmetric surfaces, and $\Phi$ vanishes exactly at two points on each $S^{2}$, the intersection of the axis $\mathcal{A}$ and $\Sigma_{0}$ is the union of two disjoint curves $\gamma_{1}$ and $\gamma_{2}$.

Next, we consider the foliation whose leaves are the orthogonal manifolds, i.e.~the manifolds whose tangent space at each point is orthogonal to the plane spanned by the vector fields $V,\Phi$. The existence (and uniqueness) of this foliation follows from our integrability assumption A4 on the spacetime geometry. We will show that, although $\Sigma_{0}$ can be foliated using 2-dimensional axisymmetric surfaces in many different ways, the integrability assumption A4 implies the existence of such a 2-dimensional foliation of $\Sigma_{0}$ with additional properties.

The intersection of the orthogonal integral submanifolds with $\Sigma_{0}$ gives rise to a 1-dimensional foliation $\mathcal{F}_{1}$ of $\Sigma_{0}$. Each leaf of this foliation is a curve whose endpoints are on $\gamma_{1}$ and $\gamma_{2}$. Indeed, the orthogonal manifolds are 2-dimensional and also dim$\Sigma_{0}=3$. Moreover, $V$ is not orthogonal to $\Sigma_{0}$ for $\epsilon$ sufficiently small and so the leaves are curves (and not 2-dimensional manifolds).  Clearly, the image of any such leaf under any diffeomorphism of the flow $\varphi^{\Phi}_{s}, 0\leq s\leq 2\pi$, of $\Phi$ is another leaf of the same foliation of $\Sigma_{0}$. Hence the above 1-dimensional foliation of $\Sigma_{0}$ gives rise to a 2-dimensional foliation $\mathcal{F}_{2}$ of $\Sigma_{0}$; the leaves of the latter foliation being the 2-dimensional images of the leaves of $\mathcal{F}_{1}$ under the flow of $\Phi$.
 
From now on we consider these 2-dimensional surfaces $S^{2}$, the leaves of $\mathcal{F}_{2}$, which we parametrize by a smooth coordinate $r$.  Clearly, the spheres $S^{2}_{r}$ are axisymmetric, i.e.~$\Phi$ is tangential to $S^{2}_{r}$. The two points on $S^{2}_{r}$ where $\Phi=0$ are called poles of $S^{2}_{r}$. We will next define an adapted coordinate system on $\Sigma_{0}$ and also show that one of these surfaces coincides with $S=\hh\cap\Sigma_{0}$. 
\begin{figure}[H]
	\centering
		\includegraphics[scale=0.18]{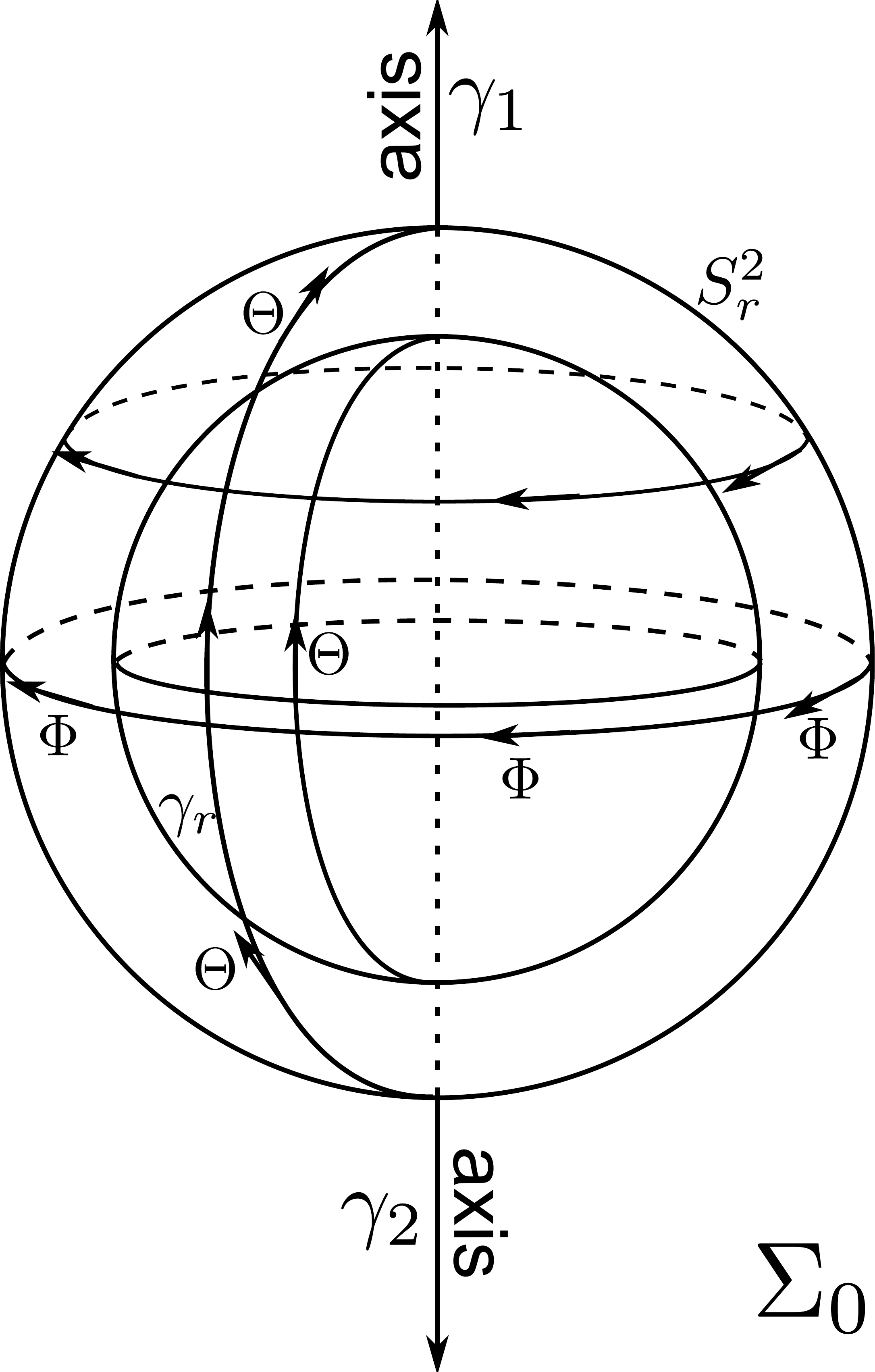}
	\label{fig:p1}
\end{figure}
Consider a leaf $\gamma_{r}$ of the foliation $\mathcal{F}_{1}$  on each  $S^{2}_{r}$ such that $\gamma_{r}$ depends smoothly on $r$. We introduce a smooth monotonic coordinate $\th$ on each $\gamma_{r}$ such that $\th=0$ at the north pole and $\th=\pi$ at the south pole of $S^{2}_{r}$. We extend $\theta$ globally on $S^{2}_{r}$ (and thus on $\Sigma_{0}$) under the condition $\lie_{\Phi}\theta=0$.
\begin{remark}
One `Hamiltonian' way to introduce the coordinate $\th$ is the following: Let  $\gi$ be the induced metric on the spheres $S_{r}^{2}$, $\epsi$ the area form and $\di$ the exterior derivative. Consider the 1-form $\a=\frac{2\pi^{2}}{A}i_{\Phi}\epsi=\frac{2\pi^{2}}{A}\epsi(\Phi,\cdot),$
where $A=\int_{S_{r}^{2}}\epsi$ is the area of $S^{2}_{r}$. By Cartan's formula we have $\di\a=\frac{2\pi^{2}}{A}\di(i_{\Phi}\epsi)=\frac{2\pi^{2}}{A}\lie_{\Phi}\epsi,$ since $\di\epsi=0$. However, $\lie_{\Phi}\epsi=0$, the vector field $\Phi$ being Killing. By virtue of the Poincar\'e lemma we obtain a globally defined function $\th$  such that $\di\th=\a$. We remark that  $\th$ is axisymmetric and has exactly two extremal points precisely at the poles $\mathcal{N},\mathcal{S}$ of $S^{2}_{r}$, and  by choosing an appropriate additive constant, $\th(\mathcal{N})=0, \ \th(\mathcal{S})=\pi$. Moreover, $g(\Theta,\Theta)=\frac{A^{2}}{4\pi^{4}}\frac{1}{g(\Phi,\Phi)}.$
\label{rehami}
\end{remark}
Let also $\phi$ denote the affine function of $\Phi$ (i.e.~$\Phi\phi=1$) such that $\phi=0$ precisely on the curves $\gamma_{r}$. Note that $\phi$ is periodic with period $2\pi$. The above construction gives rise to a coordinate system $(r,\th,\phi)$ of $\Sigma_{0}$ such that $\partial_{\phi}=\Phi$. Hence, introducing  the coordinate vector fields $Y, \Theta$ we obtain
\begin{equation*}
Y=\partial_{r},\ \ \ \ \Phi=\partial_{\phi},\ \ \ \  \Theta=\partial_{\th}.
\end{equation*}
By construction we have $V\perp \Theta$ and $\Phi\perp\Theta$ everywhere on $\Sigma_{0}$.

We now show that $S=\hh\cap\Sigma_{0}$ coincides with one of the spheres $S^{2}_{r}$.    The (orthogonal) Frobenius condition 
\begin{equation*}
\Phi_{\flat}\wedge \di\Phi_{\flat}=0
\end{equation*}
is trivially satisfied. Here $\gi$ denotes the induced metric on $S$, $\di$ the exterior derivative and  $\Phi_{\flat}$ denotes the 1-form on $S$ that corresponds to $\Phi$ under $\gi$. Hence, there exists a curve $\gamma$ on $S$ which is orthogonal to $\Phi$ and connects the north pole $\mathcal{N}$ and the south pole $\mathcal{S}$  of $S$. We consider now the 2-dimensional manifold
\begin{equation*}
\mathcal{O}=\bigcup_{t\in\rr}\varphi^{V}_{t}(\gamma),
\end{equation*}
where $\varphi^{V}_{t}$ is the flow of $V$. Since $V$ is tangential to $\hh$, we have $\mathcal{O}\subset\hh$, and moreover, $\mathcal{O}$ coincides with one of the orthogonal integral manifolds. Indeed, at each point $p\in\varphi^{V}_{t}(\gamma)$ we have
\begin{equation*} 
T_{p}\mathcal{O}=\text{span}\Big(V,d\varphi^{V}_{t}(\overset{.}{\gamma})\Big)\subset T_{p}\hh.
\end{equation*}
Hence, $V\perp T_{p}\mathcal{O}$; moreover, $\Phi\perp T_{p}\mathcal{O}$. Indeed, $g(V,\Phi)=0$ on $\hh$ and  \[g(\Phi,d\varphi^{V}_{t}(\overset{.}{\gamma}))=g(d\varphi^{V}_{t}(\Phi),d\varphi^{V}_{t}(\overset{.}{\gamma}))=g(\Phi,\overset{.}{\gamma})=0,\] by the construction of $\overset{.}{\gamma}$ on $S$.  Note that for the above equality we used that the flow of the Killing field $V$ consists of isometries and that $[V,\Phi]=0$. Hence,  the curve $\gamma$ is a leaf of the foliation $\mathcal{F}_{2}$,  $\mathcal{O}$ being an orthogonal integral submanifold, and therefore, indeed $S$, in view of its axisymmetry, coincides with one of the spheres $S^{2}_{r}$.

Let $v$ denote the affine function of the vector field $V$, namely a globally defined coordinate $v$ such that $Vv=1$ and $v=0$ on $\Sigma_{0}$ (note that, for small $\epsilon$, $V$ is transversal to $\Sigma_{0}$). By Lie-propagating $Y, \Theta, \Phi$ using the flow of $V$, we obtain the coordinate system $(v,r,\th,\phi)$ defined in a neighbourhood of the horizon $\hh$. Note that  $\partial_{v}=V$ and by virtue of $[V,\Phi]=0$ we have $\partial_{\phi}=\Phi$. Hence,
\begin{equation*}
V=\partial_{v}, \ \ \ \ Y=\partial_{r}, \ \ \ \ \Phi=\partial_{\phi}, \ \ \ \ \Theta=\partial_{\th}. 
\end{equation*}
Clearly the Lie brackets of the above vector fields vanish. \begin{figure}[H]
	\centering
		\includegraphics[scale=0.14]{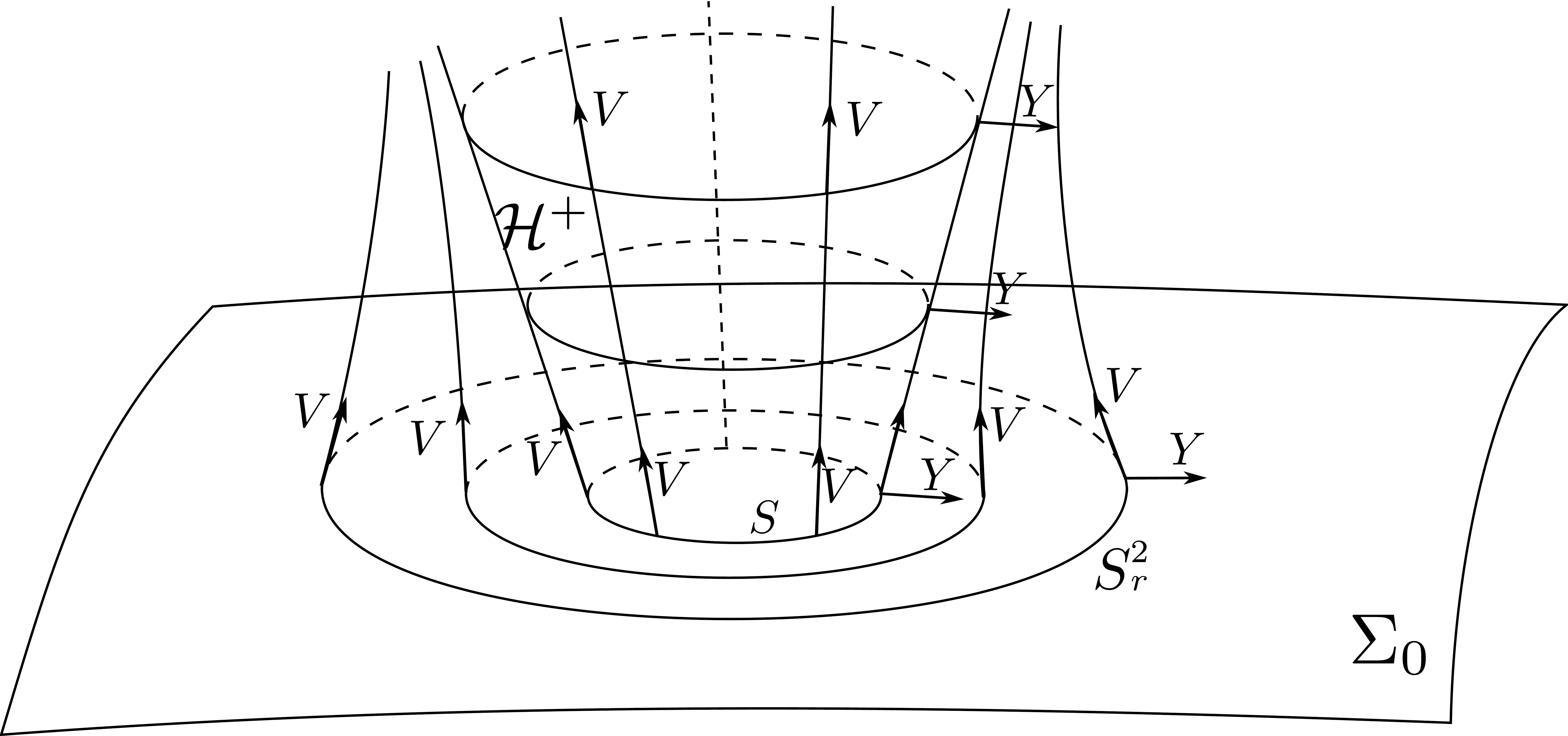}
	\label{fig:p2}
\end{figure}Note also that $\Theta\perp V$ and $\Theta\perp\Phi$, everywhere in the region under consideration. Indeed, by virtue of the fact that $V$ is Killing we have
\begin{equation*}
\lie_{V}g(\Theta,V)=g(\lie_{V}\Theta,V)+g(\Theta,\lie_{V}V)=0.
\end{equation*}
Similarly we obtain that $g(\Theta,\Phi)=0$, which completes the proof of the proposition.
\end{proof}
Let $g_{WZ}=g(W,Z)$ denote the components of the metric $g$ with respect to the coordinate basis $W,Z\in\left\{V,Y,\Theta,\Phi\right\}$ and $g^{WZ}$ denote the components of the inverse metric $g^{-1}$. Summarising, we have $ g_{V\Theta}=0, \,  g_{\Phi\Theta}=0$ \textbf{everywhere}; moreover, the following relations hold \textbf{on the horizon $\hh$}:
\begin{equation}
\begin{split}
&\!\!\!\!\!\!\!\! \!\!\!g_{VV}=0,\ \ \ g_{V\Theta}=0,\ \ \ g_{V\Phi}=0,\ \ \ g_{VY}\neq 0,\\
  \det(g)=&-\left(g_{VY}\right)^{2}\cdot g_{\Theta\Theta}\cdot g_{\Phi\cdot\Phi}=-\left(g_{VY}\right)^{2}\cdot\det(\gi),\\
&\ \ \ \   g^{YY}=0,\ \ \ g^{Y\Theta}=0, \ \ \ g^{Y\Phi}=0,\\
g^{VY}=\frac{1}{g_{VY}}&,\ \ \  g^{\Theta\Theta}=\frac{1}{g_{\Theta\Theta}},\ \ \ g^{\Phi\Phi}=\frac{1}{g_{\Phi\Phi}}, \ \ \ g^{\Theta\Phi}=0,\\
 Y(&g_{VV})=0, \ \ \ Y(g^{YY})=0, \ \ \ Y(g^{Y\Theta})=0.
\label{metric}
\end{split}
\end{equation}
The last set of equalities is a consequence of the fact that the surface gravity $\kappa$ vanishes on $\hh$ and that $\kappa= \Gamma_{VV}^{V}=-\frac{1}{2g_{VY}}Y(g_{VV})$, where $\Gamma_{VV}^{V}$ denotes the Christoffel symbol.

\begin{remark} Whenever $V$ is timelike one can obtain a local decomposition $\m=\Gamma_{1}\times\Gamma_{2}$, such that $\text{dim}\,\Gamma_{1}=\text{dim}\,\Gamma_{2}=2$ and at each point $p\in\m$ we have $T_{p}\Gamma_{1}=\text{span}\left\langle V,\Phi\right\rangle$ and $T_{p}\Gamma_{1}\perp T_{p}\Gamma_{2}$. See \cite{heusler} for more. Clearly, however, this decomposition breaks down at the points where $V$ is null and normal to $\Phi$.
\end{remark}


\subsection{The Vacuum and Electrovacuum Reduction}
\label{sec:TheVacuumAndElectrovacuumReduction}

We next put in context the relevance of the assumptions A1--A4 in the framework of General Relativity. Let $(\m,g)$ satisfy the Einstein-vacuum equations or the Einstein--Maxwell equations and the assumptions A1--A2. 

Regarding assumption A3, the null structure equations for the torsion $\eta$ (see \cite{DC09}) coupled with the extremality condition \eqref{extrem} imply that the Euler characteristic cannot be negative, and moreover, it is equal to $0$ only if the induced metric on the horizon sections is Ricci flat. It is interesting that no use of the second variation of the area is made in order to restrict the topology of extremal horizons (indeed, $\hh$ does not  have to be the future boundary of the past of null infinity as Hawking's topology theorem requires). On the other hand, Chru\'{s}ciel,  Reall and  Tod \cite{nonstaticextrem} have shown that there do not exist \textit{static} vacuum extremal black holes with spherical topology.

Regarding assumption A4, Papapetrou \cite{papapetrou} showed that the assumption $1$ and the Einstein equations (allowing for a cosmological constant $\Lambda\in\rr$) imply that the distribution of the planes orthogonal to $V,\Phi$ is integrable. The Papapetrou theorem has important applications in the black hole uniqueness problem \cite{heusler}.

Another very remarkable feature of extremal horizons is their limited dynamical degrees of freedom in the context of General Relativity. In fact, H\'{a}j\'{i}\u{c}ek \cite{hajicek} and later Lewandowski and Pawlowski \cite{extremalrigidity} and Kunduri and Lucietti \cite{luciettikund} showed a \textit{rigidity} result for extremal horizons, namely that the assumptions A1--A2 and the Einstein equations (either vacuum of electrovacumm) imply that the geometry of the horizon $\hh$ necessarily coincides with the geometry of the horizon of extremal Kerr (in the vacuum case) or extremal Kerr--Newman (in the electrovacuum case). More precisely, in the context of the formulation of \cite{DC09}, the rigidity statement for extremal horizons says that the induced metric $\gi$ of the horizon sections, the torsion $\eta$ and the curvature components $\rho,\,\sigma$ coincide with those of extremal Kerr (or Kerr--Newman). On the other hand, the transversal second fundamental form $\underline{\chi}$ is conserved (i.e.~$\lie_{V}\underline{\chi}=0$ on $\hh$) but cannot be fully determined.

\section{Conservations Law along Extremal Horizons}
\label{sec:ConservationLawAlongExtremalHorizon}

 We will prove that the degeneracy of redshift gives rise to a conservation law along $\hh$ for solutions to the wave equation \eqref{wave}. 

\subsection{The General Case}
\label{sec:TheGeneralCase}

Let $(\m,g)$ satisfy the assumptions A1--A4 of Section \ref{sec:Assumptions}. Consider the coordinate system $(v,r,\th,\phi)$  of Section \ref{sec:TheAdaptedCoordinateSystem} and let $V,Y,\Theta,\Phi$ denote the corresponding coordinate vector fields. Let $S_{0}$ be an axisymmetric section of $\hh$ and denote $S_{\tau}=d\varphi^{V}_{\tau}(S_{0})$ which are manifestly isometric to $S_{0}$. We then have the following
\begin{proposition} There exist smooth bounded  functions $\a,\b,\ga$ on $\hh$ with $\lie_{V}\a=\lie_{V}\b=\lie_{V}\ga=0$ and
such that for all solutions $\psi$ to the wave equation on $(\m,g)$ we have that the quantity 
\begin{equation}
H[\psi](\tau)=\int_{S_{\tau}}\Big(Y\psi+\a\cdot \psi+\b\cdot(V\psi)+\ga\cdot(\Theta\psi)\Big)
\label{conservation}
\end{equation}
is conserved along $\hh$, i.e.~$\lie_{V}H=0$. The above integral is taken with respect to the induced volume form on $S_{\tau}$.
\label{conservationprop}
\end{proposition}
\begin{proof}
The wave operator is given by
\begin{equation*}
\begin{split}
\Box_{g}\psi&=\text{tr}_{g}(\nabla^{2}\psi)\\&= g^{VV}\Big(VV\psi-(\nabla_{V}V)\psi\Big)+2g^{VY}\Big(VY\psi-(\nabla_{V}Y)\psi\Big)+2g^{V\Theta}\Big(V\Theta\psi-(\nabla_{V}\Theta)\psi\Big)\\&+2g^{V\Phi}\Big(V\Phi\psi-(\nabla_{V}\Phi)\psi\Big)+g^{YY}\Big(YY\psi-(\nabla_{Y}Y)\psi\Big)+2g^{Y\Theta}\Big(Y\Theta\psi-(\nabla_{Y}\Theta)\psi\Big)\\
&+2g^{Y\Phi}\Big(Y\Phi\psi-(\nabla_{Y}\Phi)\psi\Big)+g^{\Theta\Theta}\Big(\Theta\Theta\psi-(\nabla_{\Theta}\Theta)\psi\Big)+2g^{\Theta\Phi}\Big(\Theta\Phi\psi-(\nabla_{\Theta}\Phi)\psi\Big)\\
&+g^{\Phi\Phi}\Big(\Phi\Phi\psi-(\nabla_{\Phi}\Phi)\psi\Big).
\end{split}
\end{equation*}
In view of the properties \eqref{metric} of the adapted coordinate system we have
\begin{equation*}
g^{YY}=g^{Y\Theta}=g^{Y\Phi}=g^{\Theta\Phi}=0: \text{ on }\hh.
\end{equation*}
Hence, recalling \eqref{extrem} we obtain on $\hh$:
\begin{equation}
\begin{split}
\Box_{g}\psi=&\,g^{VV}\cdot(VV\psi)+2g^{VY}\Big(VY\psi-(\nabla_{V}Y)\psi\Big)+2g^{V\Theta}\Big(V\Theta\psi-(\nabla_{V}\Theta)\psi\Big)\\&+2g^{V\Phi}\Big(V\Phi\psi-(\nabla_{V}\Phi)\psi\Big)+g^{\Theta\Theta}\Big(\Theta\Theta\psi-(\nabla_{\Theta}\Theta)\psi\Big)+g^{\Phi\Phi}\Big(\Phi\Phi\psi-(\nabla_{\Phi}\Phi)\psi\Big).
\end{split}
\label{mmm}
\end{equation}
We first prove that the coefficient of $Y\psi$ in the wave operator restricted on $\hh$ vanishes. Indeed, by expanding the covariant derivatives and, in view of the \eqref{metric}, we compute the Christoffel symbols
\begin{equation*}
\begin{split}
&\Gamma_{VY}^{Y}=0, \ \ \ \ \Gamma_{V\Theta}^{Y}=0, \ \ \ \ \Gamma_{V\Phi}^{Y}=0, \ \ \ \  \Gamma_{\Theta\Theta}^{Y}= 0, \ \ \ \
\Gamma_{\Phi\Phi}^{Y}=0:\text{ on }\hh.
\end{split}
\end{equation*}
Note also that 
\begin{equation*}
\begin{split}
&\Gamma_{VY}^{\Theta}=-\frac{1}{2}g^{\Theta\Theta}\cdot\Theta(g_{VY}),\ \ \ \ \Gamma_{V\Theta}^{\Theta}=0,\ \ \ \ \Gamma_{V\Phi}^{\Theta}=0,\ \ \ \ \\&\Gamma_{\Theta\Theta}^{\Theta}=\frac{1}{2}g^{\Theta\Theta}\cdot\Theta(g_{\Theta\Theta}), \ \ \ \ \Gamma_{\Phi\Phi}^{\Theta}=-\frac{1}{2}g^{\Theta\Theta}\cdot\Theta(g_{\Phi\Phi}).  \\
\end{split}
\end{equation*}
 Hence, the coefficient of $\Theta\psi$ in \eqref{mmm} is
\begin{equation*}
\begin{split}
g^{\Theta\Theta}\Big[g^{VY}\cdot\Theta(g_{VY})-\frac{1}{2}g^{\Theta\Theta}\cdot\Theta(g_{\Theta\Theta})+\frac{1}{2}g^{\Phi\Phi}\cdot\Theta(g_{\Phi\Phi})  \Big]=\frac{1}{\sqrt{g}}\Theta\Big(\sqrt{g}\cdot g^{\Theta\Theta}\Big).
\end{split}
\end{equation*}
Therefore, there exist functions $A=A(\theta),\, B=B(\theta)$ such that 
\begin{equation*}
\begin{split}
\Box_{g}\psi=&\, V\Big(g^{VV}\cdot(V\psi)+2g^{VY}\cdot (Y\psi)+2g^{V\Theta}\cdot(\Theta\psi)+2g^{V\Phi}\cdot(\Phi\psi)+A\cdot\psi\Big)\\&+\Phi\Big(g^{\Phi\Phi}\cdot(\Phi\psi)+B\cdot\psi\Big)+\frac{1}{\sqrt{g}}\Theta\Big(\sqrt{g}\cdot g^{\Theta\Theta}\cdot(\Theta\psi)\Big)
\end{split}
\end{equation*}
\textbf{on the horizon $\hh$}.  We observe that 
\begin{equation*}
\int_{S_{\tau}}\frac{1}{\sqrt{\gi}}\Theta\Big(\sqrt{g}\cdot g^{\Theta\Theta}\cdot(\Theta\psi)\Big)=2\pi\int_{0}^{\pi}\Theta\Big(\sqrt{g}\cdot \sqrt{g^{\Theta\Theta}}\cdot(\sqrt{g^{\Theta\Theta}}\Theta)\psi\Big)d\theta=0,
\end{equation*}
since  $\sqrt{g^{\Theta\Theta}}\Theta$ is a unit vector field and 
\begin{equation*}
\sqrt{g}\cdot \sqrt{g^{\Theta\Theta}}=g_{VY}\cdot\sqrt{\gi}\cdot \sqrt{g^{\Theta\Theta}}=g_{VY}\cdot\sqrt{g_{\Phi\Phi}}\rightarrow 0
\end{equation*}
as $\theta\rightarrow 0$ or $\theta\rightarrow \pi$. Moreover  we have $\int_{S_{\tau}}\Phi f=0,$ for any sufficiently regular function $f$ on $\hh$. Therefore, by integrating $g_{VY}\cdot(\Box_{g}\psi)$ over $S_{\tau}$ and using that $g_{VY}\cdot g^{VY}=1$, we deduce that there exist smooth bounded  functions $\a,\b,\ga$ (which depend only on $\theta$) such that 
\begin{equation*}
V\int_{S_{\tau}}\Big(Y\psi+\a\cdot \psi+\b\cdot(V\psi)+\ga\cdot(\Theta\psi)\Big)=0,
\end{equation*}
which completes the proof.
\end{proof}
We remark that the extremality condition \eqref{extrem} is the key ingredient in obtaining the vanishing of the coefficient of $Y\psi$ in the wave operator (see \eqref{mmm}) on $\hh$.

\subsection{A Hierarchy of Conservation Laws in Spherical Symmetry}
\label{sec:AHierarchyOfConservationLaws}

In case the spacetime metric is spherically symmetric (i.e.~the rotation group $SO(3)$ acts on $(\m,g)$ by isometry), then the conservation law of Proposition \ref{conservationprop} can be interpreted as a conservation law for the spherical mean of $\psi$. (The spherical mean of a function is defined to be the projection of the function on the kernel of the spherical Laplacian $\lapp$. In this case, the spheres are the orbits of the $SO(3)$ action). It turns out that in spherically symmetric spacetimes one can extend this result to obtain a conservation law along $\hh$ for each of the projections of $\psi$ on the eigenspaces of $\lapp$, i.e.~a conservation law for each angular frequency of $\psi$. Indeed, the metric on such backgrounds can be expressed in the form~\cite{DC91}:
\begin{equation}
g=-D(v,r)dv^{2}+2dvdr+K^{-1}\,g_{\mathbb{S}^{2}},
\label{spmetric}
\end{equation}
where $K=K(r)$ denotes the Gaussian curvature $K$ of the spheres and  $g_{\mathbb{S}^{2}}$  the standard metric on the unit sphere $\mathbb{S}^{2}$.  As in the previous section, we denote
\[V=\partial_{v},\ \ \ Y=\partial_{r}. \]
Assuming stationarity then 
\begin{equation}
D(v,r)=D(r)=-g_{VV}.
\label{metric1}
\end{equation}
We suppose that the hypersurface $r=r_{\hh}$ is null and that $D(r_{\hh})=0$ on $\hh$. Then the extremality condition $\nabla_{V}V=0$ along  $\hh$ is incorporated in the condition 
\begin{equation}
D'(r_{\hh})=Y\big(g_{VV}\big)=0 \text{ on }\hh.
\label{spsyext}
\end{equation}
Let, moreover, $\psi_{l}$ denote the projection of $\psi$ on the eigenspace of $\lapp$ which corresponds to the eigenvalue $-l(l+1)$. If $\psi$ solves the wave equation  \eqref{wave} and $(\m,g)$ is spherically symmetric, then $\psi_{l}$ also solves the wave equation.
We now extend the Proposition \ref{conservationprop} to a hierarchy of conservation laws for spherically symmetric spacetimes:
\begin{proposition}
Let $(\m,g)$ be a spherically symmetric spacetime such that the metric $g$ satisfies the conditions \eqref{spmetric}, \eqref{metric1} and the extremality condition \eqref{spsyext}. Then for all $l\in\mathbb{N}$ there exist constants $\beta_{i},i=0,1,\dots,l,$ which depend only on $l$, such that for all solutions $\psi$ of the wave equation the quantity
\begin{equation}
H_{l}[\psi]=Y^{l+1}\psi_{l}+\sum_{i=0}^{l}\beta_{i}\cdot \big(Y^{i}\psi_{l}\big)
\label{sppropcon}
\end{equation}
is conserved along the null geodesics of $\hh$, provided the following relation holds on $\hh$:
\begin{equation}
D''(r_{\hh})=\Big.YY(g_{VV})\Big|_{r=r_{\hh}}=2K(r_{\hh}).
\label{additional}
\end{equation}
In other words, in this case, $H_{l}[\psi]$ is conserved as a function on the sections $S_{\tau}$.
\label{spsyprop}
\end{proposition}
\begin{proof}
We compute
\begin{equation*}
\begin{split}
\Box_{g}\psi=D\cdot (YY\psi)+2(VY\psi)+\beta\cdot(V\psi)+R\cdot(Y\psi)+K\cdot\big(\lapp_{\mathbb{S}^{2}}\psi\big),
\end{split}
\end{equation*}
where $\beta=\frac{YK^{-1}}{K^{-1}}$ and $R=\left[\frac{YK^{-1}}{K^{-1}}D+YD\right]$.

For $l=0$ then $\lapp_{\mathbb{S}^{2}}\psi_{0}=0$, and since $D(r_{\hh})=(YD)(r_{\hh})=0$ we obtain
\[V\Big(Y\psi+\beta_{0}\psi\Big)=0: \text{ on }\hh, \]
where $\beta_{0}=\left.\left(\frac{YK^{-1}}{2K^{-1}}\right)\right|_{r=r_{\hh}}$.
For $l\geq 1$ we need the condition \eqref{additional}. First note that 
\begin{equation}
\begin{split}
0=Y^{k}\big(\Box_{g}\psi_{l}\big)=&\ D\cdot (Y^{k+2}\psi_{l})+2Y^{k+1}V\psi_{l}+V\big(Y^{k}(\beta\psi_{l})\big)+R\cdot(Y^{k+1}\psi_{l})\\
&+\sum_{i=1}^{k}\binom{k}{i}\big(Y^{i}D\big)\cdot \big(Y^{k-i+2}\psi_{l}\big)+\sum_{i=1}^{k}\binom{k}{i}\big(Y^{i}R\big)\cdot \big(Y^{k-i+1}\psi_{l}\big)\\&-l(l+1)\cdot Y^{k}\big(K\cdot \psi_{l}\big).
\label{commutat}
\end{split}
\end{equation}
The coefficients of $Y^{k+2}\psi_{l}, Y^{k+1}\psi_{l}$ vanish on $\hh$. In view of \eqref{additional}, the coefficient of $Y^{k}\psi_{l}$ on $\hh$ is equal to
\[\binom{k}{2}\cdot D''+\binom{k}{1}\cdot R'-l(l+1)\cdot K=\binom{k+1}{2}\cdot D''-\binom{l+1}{2}\cdot D'',  \]
and, hence, is non-zero if and only if $l\neq k$. Therefore, using an inductive argument, one can easily see that for $k\leq l-1$ there exist constants $\alpha_{i}^{k},i=1,\dots,k+1,$ which depend only on $l$, such that 
\begin{equation}
Y^{k}\psi_{l}=\sum_{i=0}^{k+1}\alpha_{i}^{k}\cdot (VY^{i}\psi_{l}).
\label{last}
\end{equation} 
Applying now \eqref{commutat} for $k=l$, we have that the coefficients of $Y^{l+2}\psi_{l},Y^{l+1}\psi_{l},Y^{l}\psi_{l}$ vanish on $\hh$. Therefore, only the terms $VY^{j}\psi_{l},j=0,1,...,l+1$ and $Y^{j}\psi_{l},j=0,1,...,l-1$ remain. Applying \eqref{last} completes the proof of the conservation law for $\psi_{l},l\geq 1$.
\end{proof}

We remark that the conservation law for the spherical mean does not require the condition \eqref{additional}. Note also that extremal Reissner--Nordstr\"{o}m satisfies the condition \eqref{additional}, since $D(r)=\left(1-\frac{M}{r}\right)^{2},M>0,$ and  $K(r)=\frac{1}{r^{2}}$. These conservation laws (as well as the trapping effect on $\hh$) are in fact one of the main obstructions to obtaining the stability results of \cite{aretakis1,aretakis2}. 

In fact, we could have generalised Proposition \ref{spsyprop} so as to include more general axisymmetric extremal horizons by imposing additional geometric conditions similar to that of \eqref{additional}. However, in order not to obscure the main ideas with technicalities, we established the hierarchy of conservation laws only in the spherically symmetric setting. We will derive an analogous hierarchy of conservation laws for extremal Kerr backgrounds.

We next show non-decay and blow-up results that follow from the conservation laws.

\section{Non-Decay and Blow-up along $\hh$}
\label{sec:NonDecayAndBlowUpAlongHh}

Consider the general setting of Section \ref{sec:TheGeometricSetting}. We have the following
\begin{mytheo}\textbf{(Non-Decay)}
Let $(\m,g)$ satisfy the assumptions A1--A4 of Section \ref{sec:Assumptions} and $V,Y,\Theta,\Phi$ be the vector fields defined in Section \ref{sec:TheAdaptedCoordinateSystem}. Then, for all solutions $\psi$ to the wave equation either the non-generic condition $H[\psi]=0$  is satisfied (where $H[\psi]$ is defined in Proposition \ref{conservationprop}) or the quantity 
\[\int_{S_{\tau}}\Big(|\psi|+|Y\psi|+|\Theta\psi|+|V\psi|\Big)\]
\textbf{does not decay} along $\hh$. 
\label{nondecay} 
\end{mytheo}
\begin{proof}
Immediate from Proposition \ref{conservationprop}.
\end{proof}
 We next show that if we impose an additional geometric condition on the horizon (which holds for all known extremal black holes in general relativity), then either $\psi$ does not decay or the higher order derivatives of $\psi$ blow up along the horizon $\hh$.
\begin{mytheo}\textbf{(Blow-up)}
Assume $(\m,g)$ is as in Theorem \ref{nondecay} and moreover such that 
\begin{equation}
g(Y,\Theta)=0
\label{add}
\end{equation}
everywhere on $\hh$ and such that 
\begin{equation}
g_{VY}\cdot (YY(g^{YY}))
\label{constant}
\end{equation}
is constant and non-zero along $\hh$. Then,  unless $\psi$ and the tangential to $\hh$ derivatives of $\psi$ do not decay and $H[\psi]=0$, there is a second order derivative of $\psi$ which \textbf{blows up} along the horizon $\hh$.
\label{blowup}
\end{mytheo}
\begin{proof}
We have:
\begin{equation*}
\begin{split}
&\ \ \ \ Y\big(\sqrt{g}\cdot (\Box_{g}\psi)\big)=\\&YV(\sqrt{g}\cdot g^{VV}\cdot {V}\psi)+{YV}(\sqrt{g}\cdot g^{VY}\cdot {Y}\psi)+{YV}(\sqrt{g}\cdot g^{V\Theta}\cdot {\Theta}\psi)+{YV}(\sqrt{g}\cdot g^{V\Phi}\cdot {\Phi}\psi)\\
&{YY}(\sqrt{g}\cdot g^{YV}\cdot {V}\psi)+{YY}(\sqrt{g}\cdot g^{YY}\cdot {Y}\psi)+{YY}(\sqrt{g}\cdot g^{Y\Theta}\cdot {\Theta}\psi)+{YY}(\sqrt{g}\cdot g^{Y\Phi}\cdot {\Phi}\psi)\\
&{Y\Theta}(\sqrt{g}\cdot g^{\Theta V}\cdot {V}\psi)+{Y\Theta}(\sqrt{g}\cdot g^{\Theta Y}\cdot {Y}\psi)+{Y\Theta}(\sqrt{g}\cdot g^{\Theta\Theta}\cdot {\Theta}\psi)+{Y\Theta}(\sqrt{g}\cdot g^{\Theta\Phi}\cdot {\Phi}\psi)\\
&{Y\Phi}(\sqrt{g}\cdot g^{\Phi V}\cdot {V}\psi)+{Y\Phi}(\sqrt{g}\cdot g^{\Phi Y}\cdot {Y}\psi)+{Y\Phi}(\sqrt{g}\cdot g^{\Phi\Theta}\cdot {\Theta}\psi)+{Y\Phi}(\sqrt{g}\cdot g^{\Phi\Phi}\cdot {\Phi}\psi)\\
\end{split}
\end{equation*}
First observe that \eqref{add} implies that $YY(g^{Y\Theta})=0$ on $\hh$. Since $YY(g_{VV})\neq 0$ on $\hh$, the only terms that do not involve the $V$ or $\Phi$ derivative are precisely the following:
\begin{equation}
\begin{split}
\sqrt{g}\cdot \big(YY(g^{YY})\big)\cdot Y\psi,\ \ \ \ 
\Theta Y(\sqrt{g}\cdot g^{\Theta\Theta}\cdot \Theta\psi).
\end{split}
\label{terms}
\end{equation}
Therefore, since \[\int_{S_{\tau}}\frac{1}{\sqrt{\gi}}\cdot Y\big(\sqrt{g}\cdot(\Box_{g}\psi)\big)=0,\]
all the terms involving the $\Phi$ derivative and the second term in \eqref{terms} involving the $\Theta Y$ derivative vanish on $\hh$. Hence, by integrating along the null geodesics of $\hh$ we obtain an integral identity of the form
\begin{equation}
\int_{0}^{\tau}\int_{S_{\tau}}V\Big(f\big[\psi,D\psi,DD\psi\big]\Big)\, d\gi_{S_{\tau}} d\tau + \int_{0}^{\tau}\int_{S_{\tau}} (Y\psi)\, d\gi_{S_{\tau}} d\tau=0,  
\label{integ}
\end{equation}
where $D\psi, \, DD\psi$ are expressions of the first and second order derivatives of $\psi$, respectively. Hence, if $\psi$ and its tangential to $\hh$ derivatives decay along $\hh$ then $\int_{S_{\tau}}Y\psi\rightarrow \ell= H[\psi]$, and therefore, if $H[\psi]\neq 0$ then the second term in \eqref{integ} blows up. Applying the fundamental theorem of calculus along the null geodesics of $\hh$ for the first term in \eqref{integ} we obtain that the higher order derivatives of $\psi$ blow up asymptotically along $\hh$. 

\end{proof}

We have shown that under fairly general assumptions on extremal horizons $\hh$, scalar perturbations generically do not decay and the higher order derivatives blow up pointwise along $\hh$. The following question, however, arises: What can we say about the energy of $\psi$ on the spacelike hypersurfaces $\Sigma_{\tau}=\varphi_{\tau}^{V}(\Sigma_{0})$ that cross $\hh$?  It turns out that proving blow-up for the energy of higher order derivative of $\psi$ requires showing dispersion of $\psi$ away from $\hh$. This can, of course, only be done once other analytical features have been completely understood, such as the trapping and superradiance. See \cite{tria, lecturesMD}. For this reason, in the next section, we focus  on the fundamental extremal black holes which arise in the context of General Relativity.

\section{Scalar Instability of Extremal Black Holes in General Relativity}
\label{sec:LinearInstabilityOfExtremalBlackHolesInGeneralRelativity}

We next specialise the theory developed in previous sections to the case of vacuum and electrovacuum horizons. However, in view of our discussion in Section \ref{sec:TheVacuumAndElectrovacuumReduction}, any degenerate electrovacuum horizon must be isometric to the event horizon of a member of the extremal Kerr--Newman family. Moreover, Chru\'{s}ciel and Tod \cite{chrus} proved that all static electrovacuum black hole spacetimes are isometrically diffeomorphic to the Reissner--Nordstr\"{o}m or the standard Majumdar--Papapetrou spacetime.

The case of extremal Reissner--Nordstr\"{o}m was treated in \cite{aretakis1,aretakis2}. For reference, if $M>0$ denotes the mass parameter, then the quantity \[H_{0}^{\text{RN}}[\psi]=Y\psi+\frac{1}{M}\psi\]
is conserved along $\hh$ for all spherically symmetric solutions $\psi$ to the wave equation.

We next consider the case of Majumdar--Papapetrou and extremal Kerr spacetimes. 

\subsection{Majumdar--Papapetrou Multi Black Holes}
\label{sec:MajumdarPapapetrouMultiBlackHoles}

The Majumdar--Papapetrou multi black hole spacetimes (see \cite{hartlehaw72}) constitute a family of solutions to the Einstein--Maxwell equations with $N$ extremal black holes, for some $N\in\mathbb{N}$. The mass $M_{i}$ enclosed by a section of $\hh_{i}$  is equal to the charge $e_{i}$ inside the same surface, i.e., the black holes remain in equilibrium by the consequent balance of their electrostatic repulsion and gravitational attraction.

The Majumdar--Papapetrou spacetimes are static  with defining Killing vector field $V$; the vector field $V$ being normal on each of the degenerate event horizons $\hh_{i}$, $i=1,...,N$ (and thus $\nabla_{V}V=0$ on $\hh_{i}$) and timelike in the exterior region. These spacetimes are not spherically symmetric or even axisymmetric; however, each of the horizons $\hh_{i}$ is spherically symmetric. In fact, for each $\hh_{i}$, there is a coordinate system $(v,r_{i},\theta,\phi)\in\mathbb{R}\times (-\epsilon,\epsilon)\times (0,\pi)\times (0,2\pi)$ covering a neighbourhood of $\hh_{i}$ such that $\hh_{i}=\left\{r_{i}=0\right\}$ and $\partial_{v}=V$ (see \cite{hartlehaw72}). Denoting  
\[Y=\partial_{r_{i}},\ \ \Theta=\partial_{\theta},\ \ \Phi=\partial_{\phi}\]
we also have
\begin{equation*}
\Theta\perp\ V,\ \Theta\perp Y, \ \Phi\perp V, \ \Phi\perp Y
\end{equation*}
everywhere in the domain of the system. Furthermore, $\Theta,\Phi$ are  tangential to sections $S_{\tau}$ of $\hh_{i}$ and Killing on $\hh_{i}$. Moreover, we have
\begin{equation*}
\lie_{Y}\gi=0, \ \lie_{Y}\text{det}(g)=0:\text{ on }\hh_{i},
\end{equation*}
where $\gi$ denotes the induced metric on the $(\theta,\phi)$ spheres. The above properties of the metric allow us to apply our framework and hence we conclude that there exists a bounded function $\alpha$ such that $\lie_{V}\alpha=0$ and such that the quantity
\[H_{0}^{\text{MP}}[\psi](\tau)=\int_{S_{\tau}}\Big(Y\psi+\alpha\cdot\psi\Big)\]
is conserved along $\hh$. Hence, by virtue of the above properties, Theorems \ref{nondecay} and \ref{blowup} hold for the Majumdar--Papapetrou spacetimes.

Obtaining, however, dispersive estimates (even away from $\hh_{i}$ for all $i=1,...,N$) remains an open problem.

\subsection{Extremal Kerr}
\label{sec:ExtremalKerr}

 In view of the results discussed in Section \ref{sec:TheVacuumAndElectrovacuumReduction}, extremal Kerr satisfies all the assumptions of Section \ref{sec:Assumptions}. Of course, one could verify this by inspection of the metric, which in ingoing Eddington--Finkelstein coordinates $(v,r,\theta,\phi^{*})$ takes the form
\begin{equation*}
g=g_{vv}dv^{2}+2g_{v\phi^{*}}dvd\phi^{*}+g_{\phi^{*}\phi^{*}}(d\phi^{*})^{2}+g_{\theta\theta}d\theta^{2}+2g_{vr}dvdr+2g_{r\phi^{*}}drd\phi^{*},
\end{equation*}
where
\begin{equation}
\begin{split}
&g_{vv}=-\left(1-\frac{2Mr}{\rho^{2}}\right), \  \ \  g_{\phi^{*}\phi^{*}}=\frac{(r^{2}+a^{2})^{2}-a^{2}\Delta\sin^{2}\theta}{\rho^{2}}\sin^{2}\theta, \  \ \  g_{\theta\theta}=\rho^{2}
\\&   g_{vr}=1, \ \ \   g_{v\phi^{*}}=-\frac{2M^{2}r\sin^{2}\theta}{\rho^{2}}, \  \ \  g_{r\phi^{*}}=-M\sin^{2}\theta,
\end{split}
\label{edi}
\end{equation}
with $M>0$ a constant and 
\begin{equation}
\Delta=(r-M)^{2}, \ \ \ \ \ \rho^{2}=r^{2}+M^{2}\cos^{2}\theta. 
\label{basic}
\end{equation}
The event horizon $\hh$ corresponds to $r=M$. For completeness, we include the computation for the inverse of the metric in $(v,r,\theta,\phi^{*})$ coordinates:
\begin{equation*}
\begin{split}
&g^{vv}=\frac{M^{2}\sin^{2}\theta}{\rho^{2}}, \ \ \   g^{rr}=\frac{\Delta}{\rho^{2}},\ \ \   g^{\phi^{*}\phi^{*}}=\frac{1}{\rho^{2}\sin^{2}\theta}, \  \ \  g^{\theta\theta}=\frac{1}{\rho^{2}}
\\& \ \ \ \ \ \ \  \ \   g^{vr}=\frac{r^{2}+M^{2}}{\rho^{2}} ,\  \ \   g^{v\phi^{*}}=\frac{M}{\rho^{2}}, \ \ \   g^{r\phi^{*}}=\frac{M}{\rho^{2}},
\end{split}
\end{equation*}
The metric is indeed stationary and axisymmetric and the corresponding quantity \eqref{constant} is constant. Note that the vector field $\partial_{v}$ is not null on the horizon $\hh$. If we denote $T=\partial_{v},\, Y=\partial_{r},\,\Theta=\partial_{\theta},\,\Phi=\partial_{\phi}$, then  the vector field $V=T+\frac{1}{2M}\Phi$ is null and normal to $\hh$. It follows that the Papapetrou condition A4 is satisfied.

 The wave operator is given by
\begin{equation}
\begin{split}
\Box_{g}\psi=&\frac{M^{2}}{\rho^{2}}\sin^{2}\theta\left(TT\psi\right)+\frac{2(r^{2}+M^{2})}{\rho^{2}}\left(TY\psi\right)+\frac{\Delta}{\rho^{2}}(YY\psi)\\&+\frac{2M^{2}}{\rho^{2}}(T\Phi\psi)+\frac{2M}{\rho^{2}}(Y\Phi\psi)+\frac{2r}{\rho^{2}}(T\psi)+\frac{\Delta'}{\rho^{2}}(Y\psi) +\frac{1}{\rho^{2}}\lapp_{(\theta,\phi^{*})}\psi,
\label{operator}
\end{split}
\end{equation}
where $\lapp_{(\theta,\phi^{*})}\psi$ denotes the standard Laplacian on $\mathbb{S}^{2}$ with respect to $(\theta,\phi^{*})$. If $S_{\tau}$ are the sections $v=\tau$ of $\hh$ we then obtain that the quantity
\begin{equation}
H_{0}^{\text{Kerr}}[\psi](\tau)=\int_{S_{\tau}}\Big(M\sin^{2}\th\,(T\psi)+4M\,(Y\psi)+2\psi\Big)
\label{conskerr}
\end{equation}
is conserved along $\hh$ (note that the volume form of $S_{\tau}$ is $2M\sin\theta\, d\theta\, d\phi^{*}$). In fact, we can also obtain a hierarchy of conservation laws analogous to that of Proposition \ref{spsyprop} for the spherically symmetric case. Let $E_{l}=E_{l}(\theta,\phi^{*})$ be an eigenfunction of $\lapp$ with corresponding eigenvalue equal to $-l(l+1)$. By restricting $Y(\rho^{2}\Box_{g}\psi)=0$ on $\hh$, multiplying with $E_{1}$ and using Stokes' theorem we obtain 
\begin{equation*}
\begin{split}
\int_{S_{\tau}}\!\!\bigg(\!\Big[M^{2}\sin^{2}\theta\left(TTY\psi\right)+4M^{2}\left(TYY\psi\right)+2M(TY\psi)+2(Y\psi)\Big]\!\cdot\! E_{1} +(Y\psi)\cdot \lapp E_{1}\bigg)\!=\!0.
\end{split}
\end{equation*}
Since $\lapp E_{1}=-2E_{1}$, we end up only with terms which involve the $T$ derivative. Hence, the quantity 
\[H_{1}^{\text{Kerr}}[\psi](\tau)=\int_{S_{\tau}}\bigg(\Big[4M\left(YY\psi\right)+M\sin^{2}\theta\left(TY\psi\right)+2(Y\psi)+2\psi\Big]\cdot E_{1}\bigg)  \]
is conserved along $\hh$. Similarly we obtain the following
\begin{proposition}
Let $l\in\mathbb{N}$ and $E_{l}$ denote a spherical harmonic such that $\lapp E_{l}=-l(l+1)E_{l}$. For any function $f$ we also denote  $f_{l}=f\cdot E_{l}$. There exist constants $\alpha,\,\beta_{i},i=0,1,...,l+1$ which depend only on $M$ such that for all solutions $\psi$ to the wave equation on extremal Kerr backgrounds, the quantity
\[ H_{l}^{\text{Kerr}}[\psi](\tau)=\int_{S_{\tau}}\bigg(Y^{l+1}\psi_{l}+\alpha\cdot\Big(\sin^{2}\th \cdot TY^{l}\psi\Big)_{l}+\sum_{i=0}^{l}\beta_{i}\,Y^{i}\psi_{l}  \bigg) \]
is conserved along $\hh$, i.e.~it is independent of $\tau$.
\label{hierkerr}
\end{proposition}

In order to obtain definitive instability results for extremal Kerr we use the results of \cite{aretakis3} which we recall below:

1. \textbf{Pointwise decay for $\psi$}: For all axisymmetric solutions $\psi$ which arise from regular initial data we have  $\left\|\psi\right\|_{L^{\infty}(\Sigma_{\tau})}\rightarrow 0$ as $\tau\rightarrow+\infty$. Here $\Sigma_{\tau}=\varphi_{\tau}^{V}(\Sigma_{0})$ and $\Sigma_{0}$ is a spacelike hypersurface which crosses $\hh$.

2. \textbf{Decay of degenerate (at $\hh$) energy of $\psi$}: For all axisymmetric solutions $\psi$ which arise from regular initial data we have $E_{\Sigma_{\tau}\cap\left\{r\geq r_{0}>M\right\}}[\psi]\rightarrow 0$ as $\tau\rightarrow+\infty$. 

Combining the methods of \cite{aretakis2}, the results of the present paper (note, in particular, that the condition \eqref{add} holds on extremal Kerr) and \cite{aretakis3} and by projecting to the zeroth azimuthal frequency we obtain the following
\begin{mytheo}\textbf{(Scalar Instability of Extremal Kerr)}
There exists a constant $c>0$ which depends only on $M$  such that for all solutions to the wave equation on extremal Kerr we have 
\begin{enumerate}
	\item \textbf{Non-decay}: \[\underset{S_{\tau}}{\textrm{sup}}\ \big|Y\psi\big|\ \geq \ c\big|H_{0}[\psi]\big|,\]
	along $\hh$ and $H_{0}[\psi]$ is a constant which depends only on the initial data and is generically non-zero.
	\item  \textbf{Pointwise blow-up}: \[\underset{S_{\tau}}{\textrm{sup}}\ \big|Y^{k}\psi\big|\ \geq\  c\big|H_{0}[\psi]\big|\tau^{k-1}, \]
	asympotically along $\hh$ for all $k\geq 2$.	 
\item \textbf{Energy blow-up}: For generic solutions $\psi$ to the wave equation we have 
\[E_{\Sigma_{\tau}}[\psi]=\int_{\Sigma_{\tau}}\left\langle J^{N}[Y^{k}\psi], n_{\Sigma_{\tau}} \right\rangle \ \, \longrightarrow +\infty,\] for all $k\geq 2$, as $\tau\rightarrow+\infty$. Here $J^{N}$ is the natural energy current and $n_{\Sigma_{\tau}}$ is the unit normal to $\Sigma_{\tau}$.
\end{enumerate}

\label{kerrinsta}
\end{mytheo}

\section{Acknowledgements}
\label{sec:Acknowledgements}

I would like to thank Mihalis Dafermos for his invaluable help and insights. I would also like to thank Igor Rodnianski, Amos Ori and Gustav Holzegel for several very stimulating discussions. Special thanks go to James Lucietti and Volker Schlue for carefully reading the original manuscript and suggesting several improvements.

\section{Addendum}
\label{sec:Addendum}

Since the first appearance of the present paper on the arXiv there have been rapid developments regarding instabilities of extremal black holes. These developments were motivated by the analysis of the present paper and for this reason we summarize some of these contributions below:

\medskip

1. Lucietti and Reall \cite{hj2012} have extended the conservation laws of the present paper to electromagnetic and linearized gravitational perturbations on extremal Kerr backgrounds. 

\medskip

2. Murata \cite{murata2012} has generalized the conservation laws for scalar, electromagnetic and linearized gravitational perturbations on extremal horizons in vacuum in  arbitrary dimensions. 

\medskip

3. Bizon and Friedrich \cite{bizon2012} have shown that the conservation laws of the present paper on exactly extremal Reissner--Nordstr\"{o}m  correspond to the Newman-Penrose constants at null infinity under a conformal transformation of the background which exchanges the (future) event horizon with (future) null infinity. The authors made similar comments for the conservation laws on extremal Kerr. 

\medskip

4. The relation between the conserved quantities of the present paper and the Newman--Penrose constants was also independently observed by Lucietti, Murata, Reall and Tanahashi \cite{hm2012}. The same authors studied analytically and numerically the late time behavior of massive and massless scalars on extremal Reissner--Nordstr\"{o}m. An important conclusion of their numerical analysis is that scalar instabilities are present even if the scalar perturbation is initially supported away from the horizon (in which case all the conserved quantities are zero). 

\medskip

5. The author  rigorously showed in \cite{aretakis2012} that  perturbations which are initially supported away from the horizon  indeed (generically) develop instabilities in the future confirming the numerical analysis of \cite{hm2012}. 

\medskip

6. Using the results of the present paper, the author has shown in \cite{aretakis2013} that extremal black holes exhibit a genuine scalar \textit{non-linear} instability which is not present for subextremal black holes or the Minkowski spacetime.

\bibliographystyle{acm}
\bibliography{../../../bibliography}

\end{document}